\PassOptionsToPackage{prologue,usenames,dvipsnames,table}{xcolor}
\PassOptionsToPackage{bookmarks,unicode,colorlinks=true}{hyperref}
\newif\ifanonymous

\anonymousfalse

\documentclass[runningheads]{llncs}

\usepackage{amssymb,amsmath,mathtools}
\usepackage{nicefrac}
\usepackage{mathrsfs}
\usepackage{xspace}
\usepackage[inline]{enumitem}
\usepackage{comment}
\usepackage{thm-restate}
\usepackage{graphicx}
\usepackage[hyperfootnotes=true,colorlinks,linkcolor=RedViolet,anchorcolor=RedViolet,citecolor=blue,urlcolor=black]{hyperref}
\usepackage[capitalize,nameinlink]{cleveref}
\usepackage{adjustbox}
\usepackage{bm}
\usepackage{utfsym}
\usepackage{bussproofs}
\usepackage{latexsym} 
\usepackage{tcolorbox}
\usepackage{booktabs}
\usepackage{wrapfig}
\usepackage{makecell}
\usepackage{multirow}
\usepackage{bbding,pifont}
\usepackage[misc]{ifsym}
\usepackage{orcidlink}
\usepackage{caption}
\usepackage{algorithm}
\usepackage{algpseudocodex}
\usepackage{ifthen}

\spnewtheorem{subproblem}{Problem~1.\!\!}{\itshape}{}

\Crefname{problem}{Problem}{Problem}
\Crefname{subproblem}{Problem~1.\!\!}{Problem~1.\!\!}
\crefname{lemma}{Lem.}{Lem.}
\crefname{example}{Exmp.}{Exmp.}
\crefname{section}{Sect.}{Sect.}
\Crefname{appendix}{Appx.}{Appx.}
\crefname{appendix}{App.}{App.}
\crefname{definition}{Def.}{Def.}
\crefname{theorem}{Thm.}{Thm.}
\crefname{corollary}{Cor.}{Cor.}
\crefname{algorithm}{Alg.}{Alg.}

\spnewtheorem{benchmark}{Benchmark}{\itshape}{\rmfamily}

\input{macros.tex}
\input{tikzkey.tex}

\makeatletter
\def\thanks#1{\protected@xdef\@thanks{\@thanks
        \protect\footnotetext{#1}}}
\makeatother

\begin{document}
\title{Exact Moment Estimation of Stochastic Differential Dynamics}

\titlerunning{Exact Moment estimation of SDE}
%

\ifanonymous
    \author{}
    \institute{}
\else
    \author{
    Shenghua Feng \inst{1,2}\orcidlink{0000-0002-5352-4954} \and
    Jie An \inst{1,2}\orcidlink{0000-0001-9260-9697} \and 
    Naijun Zhan \inst{3,4}\orcidlink{0000-0003-3298-3817} \and
    Fanjiang Xu \inst{1,2}\orcidlink{0009-0004-6016-7360}
    }
    \authorrunning{S. Feng et al.}

    \institute{
    National Key Laboratory of Space Integrated Information System, Institute of Software Chinese Academy of Sciences, Beijing, China \and
    University of Chinese Academy of Sciences, Beijing, China \and
    School of Computer Science, Peking University, Beijing, China \and
    Zhongguancun Laboratory, Beijing, China \\
    \email{\{fengshenghua,anjie,fanjiang\}@iscas.ac.cn} \quad \email{njzhan@pku.edu.cn}
    }
\fi

\maketitle              
\setcounter{footnote}{0}

\begin{abstract}

Moment estimation for stochastic differential equations (SDEs) is fundamental to the formal reasoning and verification of stochastic dynamical systems, yet remains challenging and is rarely available in closed form. In this paper, we study time-homogeneous SDEs with polynomial drift and diffusion, and investigate when their moments can be computed exactly. We formalize the notion of moment-solvable SDEs and propose a generic symbolic procedure that, for a given monomial, attempts to construct a finite-dimensional linear ordinary differential equation (ODE) system governing its moment, thereby enabling exact computation. We introduce a syntactic class of pro-solvable SDEs, characterized by a block-triangular structure, and prove that all polynomial moments of any pro-solvable SDE admit such finite ODE representations. This class strictly generalizes linear SDEs and includes many nonlinear models. Experimental results demonstrate the effectiveness of our approach.

\end{abstract}


\section{Introduction}\label{sec:introduction}

Stochastic differential equations (SDEs) are foundational mathematical models for describing the evolution of stochastic systems across a wide range of fields, from physics and biology~\cite{panik2017stochastic} to finance~\cite{Black+Scholes/1973/Pricing} and engineering~\cite{Hoogendoorn+Others/2004/Pedestrian}. Moments, defined as expectations of monomial functions of the system variables, provide critical insights into the behavior and stability of the underlying stochastic systems. Consequently, moment estimation is a central problem in various applications, including formal verification of stochastic dynamical systems~\cite{ghusinga2020moment,lamperski2018analysis,bartocci2019automatic}, sensitivity analysis~\cite{armstrong2021sensitivity,gunawan2005sensitivity}, and the derivation of bounds for safety-critical models~\cite{scarciotti2021moment,moosbrugger2022moment}.




Despite their importance, exact moment estimation, i.e., obtaining closed-form expressions for moments of SDEs, remains a formidable challenge.
Since SDEs, even linear ones, are often difficult to solve in closed form, their moments are correspondingly hard to compute directly. To the best of our knowledge, existing approaches based on the synthesis of supermartingales~\cite{williams1991probability} can provide bounds on moments~\cite{prajna2007framework}, but these methods are often conservative and rely heavily on the choice of template functions for the supermartingale, which may require significant manual effort. As illustrated in our case study (cf.~\cref{subsec:case_study}), such approaches may fail to yield tight estimates for the moments of interest. Currently, there is no general framework for the exact computation of moments in general SDEs.

In this paper, we focus on time-homogeneous SDEs with polynomial drift and diffusion terms, and seek to systematically characterize when the moments can be computed exactly. Motivated by advances in the verification of probabilistic programs~\cite{bartocci2019automatic}, we introduce the notion of \emph{moment-solvable} SDEs, for which every moment admits an explicit, closed-form solution. Central to our approach is a symbolic procedure that, given any monomial, attempts to construct a finite-dimensional linear ODE system governing the moment dynamics. This procedure iteratively expands the set of coupled moments by applying the infinitesimal generator of the SDE, and halts if the expansion closes after finitely many steps.

Our main theoretical contribution is the identification of a syntactic class of SDEs, termed \emph{pro-solvable SDEs}, which are characterized by a block-triangular structure in their coefficients. We prove that for all pro-solvable SDEs, the symbolic closure procedure always terminates, ensuring that every moment can be computed by solving a finite-dimensional linear ODE system. Notably, this class strictly generalizes the linear SDEs, encompassing a broad array of nonlinear models encountered in practice.


We demonstrate the practical effectiveness of our approach through experiments on a diverse suite of SDE benchmarks. Our method efficiently computes exact moments for many linear and nonlinear systems of interest, including higher-order cases, illustrating its scalability and applicability to the formal analysis of stochastic systems.

In summary, our main contributions are as follows:
\begin{itemize}
\item After the problem formulation (\cref{sec:problem}), we introduce the concept of moment-solvable SDEs and establish a general procedure for constructing finite-dimensional ODE systems governing the moment dynamics. (\cref{sec:moment-solvable-SDE})
\item We introduce and characterize the class of pro-solvable SDEs, proving that they are moment-solvable and that all their moments can be computed exactly. Furthermore, we provide a complexity analysis of our method. (\cref{sec:solvable_SDEs})
\item We implement our method and conduct experiments to demonstrate the broad applicability and effectiveness of our approach. (\cref{sec:experiments})
\end{itemize}


\myparagraph{Related work}
The closest related works are those on \emph{prob-solvable loops} and related classes of probabilistic programs in discrete time~\cite{bartocci2019automatic,moosbrugger2022moment}, where loop moments satisfy solvable linear recurrences and can be computed exactly. Our setting differs in that we consider continuous-time SDEs and build on the infinitesimal generator.
A second line of work provides structural exact-moment results for stochastic reaction networks~\cite{lee2009moment} and jump Markov processes~\cite{sontag2015exact,borri2020cubification}, where specific network topologies (i.e., feedforward structures) yield closed finite moment equations. These results, however, are tailored to particular classes of jump processes~\cite{guidoum2023exact,sontag2018examples} and do not offer a uniform procedure for general polynomial SDEs.
Finally, existing works employ martingales~\cite{prajna2007framework,feng2020unbounded,hafstein2018lyapunov} and semidefinite relaxations~\cite{ghusinga2017approximate,lasserre2018moment} to approximate moments and provide bounds used for the verification and analysis of SDEs and stochastic hybrid systems.



\section{Problem Formulation}\label{sec:problem} 

Let $\Nats$, $\Ints$, and $\Reals$ denote the sets of natural numbers, integers, and real numbers, respectively. Vectors are denoted in bold; for $\xx \in \Reals^n$, $x_i$ refers to its $i$-th component. For $\xx = (x_1,\dots,x_n)^\top$, let $\alpha \defeq (\alpha_1,\dots,\alpha_n)\in\Nats^n$ be a multi-index, with $|\alpha| \defeq \sum_{i=1}^n \alpha_i$. We use the notation $\xx^\alpha \defeq x_1^{\alpha_1}\cdots x_n^{\alpha_n}$ for presenting the corresponding monomial.

\myparagraph{Probability and Moments}
Let $(\Omega, \mathcal{F}, P)$ be a probability space, where $\Omega$ is the sample space, $\mathcal{F} \subseteq 2^\Omega$ is a $\sigma$-algebra, and $P\colon \mathcal{F} \to [0, 1]$ is a probability measure. A \emph{random variable} $X$ defined on $(\Omega, \mathcal{F}, P)$ is an $\mathcal{F}$-measurable function $X\colon \Omega \to \Reals^n$; its expectation (w.r.t. $P$) is denoted by $\EE[X]$. Given a multi-index $\alpha = (\alpha_1,\dots,\alpha_n)\in\Nats^n$, the $\alpha$-moment of $X$ is $\EE[X^\alpha] \defeq \EE[(X^{(1)})^{\alpha_1} \cdots (X^{(n)})^{\alpha_n}]$, where $X^{(i)}$ denotes the $i$-th component of $X$. A (continuous-time) \emph{stochastic process} is a collection of random variables $\{X_t\}_{t \in T}$, where unless otherwise noted, the index set $T$ is the half-line $[0, \infty)$.

\myparagraph{Stochastic Differential Equations (SDEs)}
We consider a class of stochastic dynamical systems governed by time-homogeneous stochastic differential equations (SDEs) of the form
\begin{equation}\label{eq:sde}
    dX_t = b(X_t)\,dt + \sigma(X_t)\,dW_t, \quad t \ge 0,
\end{equation}
where $\{X_t\}$ is an $n$-dimensional continuous-time stochastic process, $\{W_t\}$ is an $m$-dimensional Wiener process (standard Brownian motion), $b\colon \Reals^n \to \Reals^n$ is a vector-valued polynomial drift coefficient modeling the deterministic part of the dynamics, and $\sigma\colon \Reals^n \to \Reals^{n \times m}$ is a matrix-valued polynomial diffusion coefficient encoding the system's coupling to Gaussian white noise $dW_t$.

Under standard regularity and growth conditions~\cite[Chap.~5.2]{oksendal2013stochastic}, given an initial state (random variable) $X_0$, the SDE~\eqref{eq:sde} admits a unique solution $X_t(\omega) = X(t, \omega)\colon [0, \infty) \times \Omega \to \Reals^n$ that satisfies the stochastic integral equation
\[
X_t = X_0 + \int_0^t b(X_s)\,ds + \int_0^t \sigma(X_s)\,dW_s.
\]
The solution process $\{X_t\}$ of~\eqref{eq:sde} is also called an \emph{(It\^{o}) diffusion process}, and may be denoted $X_t^{0, X_0}$ (or simply $X_t^{X_0}$) to indicate the initial condition $X_0$ at time $t=0$. In the special case where $\sigma \equiv 0$, the SDE reduces to an ordinary differential equation (ODE), recovering the classical deterministic setting.

The \emph{exact moment estimation} (EME) problem of SDEs studied in this paper reads as follows:
\begin{tcolorbox}[boxrule=1pt,colback=white,colframe=black!75, left = 1pt, right = 1pt, top = 0pt, bottom = 0pt]
\noindent\textbf{EME Problem.} 
Let $\{X_t\}_{t\ge 0}$ denote the solution of SDE~\eqref{eq:sde}. Our objective is to compute the $\alpha$-moment of the random variable $X_t$, that is,
\[
m_\alpha(t) \ddefeq \EE[X_t^\alpha]
= \EE\big[\big(X_t^{(1)}\big)^{\alpha_1} \cdots \big(X_t^{(n)}\big)^{\alpha_n}\big]
\]
for any given multi-index $\alpha = (\alpha_1, \dots, \alpha_n) \in \Nats^n$ and any $t \ge 0$.
\end{tcolorbox} 



\section{Reduction of EME Problem to ODE Solving} \label{sec:moment-solvable-SDE}
In this section, we formalize the concept of \emph{moment-solvable} SDEs, which precisely delineates those systems for which the EME problem admits an explicit solution. We then present a generic symbolic procedure that, given a target monomial, systematically attempts to construct a finite-dimensional linear ODE system governing the evolution of its moment, thus enabling exact computation of the desired expectation.

\begin{definition}[Moment-solvable SDE]\label{def:moment-solvable}
Given a multi-index $\alpha \in \Nats^n$, we say that the SDE~\eqref{eq:sde} is \emph{moment-solvable for $\alpha$} if there exists an explicit function
$h_\alpha \colon [0, \infty) \to \mathbb{R}$ such that
\[
m_\alpha(t) = \EE[X_t^\alpha] = h_\alpha(t)
\quad\text{for all } t \ge 0.
\]
If this property holds for all multi-indices $\alpha$, then the SDE is \emph{moment-solvable}.
\end{definition}

Clearly, solving the EME problem for a multi-index $\alpha$ is equivalent to establishing that the SDE is moment-solvable for $\alpha$. In practice, obtaining explicit formulas for moments is highly nontrivial. First, closed-form solutions for nonlinear SDEs are generally unavailable. Second, even when a stochastic representation of the solution is known, the evaluation of $\EE[X_t^\alpha]$ typically involves high-dimensional integrals that do not admit simple analytical expressions.

To circumvent these challenges, rather than attempting to compute $\EE[X_t^\alpha]$ directly, we instead consider the time evolution of moments. By deriving differential equations for moments using the infinitesimal generator, we seek to construct a closed, finite-dimensional linear ODE system for the evolution of a suitable collection of moments. This approach is grounded in Dynkin’s formula, which serves as the stochastic analogue of the Newton–Leibniz rule and connects the dynamics of the process to the evolution of expected values.

\begin{theorem}[Dynkin’s formula~\cite{oksendal2013stochastic}]\label{thm:dykin}
Let \(\{X_t\}_{t\ge 0}\) be the solution of~\eqref{eq:sde}.
If \(f \in C^{2}(\mathbb{R}^n)\) has compact support, then for all \(t\ge 0\),
\[
\frac{\mathrm{d}}{\mathrm{d}t}\EE[f(X_t)] = \EE[(\mathcal{A}f)(X_t)],
\]
where \(\mathcal{A}\) is the infinitesimal generator of~\eqref{eq:sde} given by
\[
\mathcal{A}f(\xx)
= \sum_{i=1}^n b_i(\xx)\frac{\partial f}{\partial x_i}(\xx)
+ \frac{1}{2}\sum_{i,j=1}^n (\sigma\sigma^\top)_{ij}(\xx)\,
  \frac{\partial^2 f}{\partial x_i\partial x_j}(\xx).
\]
\end{theorem}

\begin{remark}
The compact-support assumption in \cref{thm:dykin} is primarily a technical requirement. By a standard localization argument, Dynkin’s formula can be extended to many unbounded functions under mild integrability/growth conditions. Specifically, the formula applies to polynomials whenever the corresponding moments exist; that is, if the relevant moments are finite, Dynkin’s formula remains valid for monomials. See Appendix~\ref{appendix:extension_dynkin} for a formal derivation.
Notably, the assumption of moment existence is appropriate here, as our primary objective is the computation of these moments. Such existence can be verified using a ranking supermartingale or a Lyapunov function, which are standard approaches in the literature (see, e.g., \cite[Chap. 1, Sect. 1.4]{khasminskii2011stochastic} and \cite[Chap. 11]{meyn2012markov}).
\end{remark}

Since both the drift vector $b(\xx)$ and the diffusion matrix $\sigma(\xx)$ are polynomial, applying Dynkin's formula to the test function $f(\xx) = \xx^\alpha$ yields
\begin{align}
    \frac{d}{dt} & m_\alpha(t) \eeq \mathbb{E}[\mathcal{A}(\mathbf{x}^\alpha)|{\mathbf{x}=X_t}]  
    \TAG{Suppose $\displaystyle \mathcal{A}(\xx^\alpha) = \sum_\gamma a_{\alpha\gamma} \xx^\gamma + c_\alpha$ } \\ 
   & \eeq  \EE\left[\sum_\gamma a_{\alpha\gamma} X_t^{\gamma}\right] + c_\alpha
    = \sum_\gamma a_{\alpha\gamma} \mathbb{E}[X_t^{\gamma}] + c_\alpha = \sum_\gamma a_{\alpha\gamma} m_{\gamma}(t) + c_\alpha \label{eq:moment_derivative}
\end{align}
This procedure reveals that the dynamics of any single moment are coupled to the dynamics of other moments. We can systematically uncover the full set of coupled moments by starting with our initial moment $m_\alpha$ and recursively applying the generator $\mathcal{A}$ to any new monomials that appear on the right-hand side of \cref{eq:moment_derivative}.

\begin{algorithm}[!t]
  \caption{Construction of a finite moment system for a given multi-index $\alpha$}
  \label{alg:moment-closure}
  \begin{algorithmic}[1]
    \Require Drift $b$, diffusion $\sigma$, generator $\mathcal{A}$ of~\eqref{eq:sde},
             initial multi-index $\alpha$
    \Ensure Multi-index set $S$ which entails a closed linear ODE of moments if terminates.
    \State $S \gets \{\alpha\}$ \Comment{set of multi-indices (monomials)}
    \State $\mathcal{P} \gets \{\alpha\}$ \Comment{set of unprocessed multi-indices}
    \While{$\mathcal{P} \neq \emptyset$}
      \State Select and remove some $\beta \in \mathcal{P}$
      \State Compute the generator action on the monomial $x^\beta$:
      $
        \mathcal{A} \xx^\beta = \sum_{\gamma} a_{\beta\gamma}\,\xx^\gamma + c_\beta
      $
      \ForAll{$\gamma$ such that $a_{\beta\gamma} \neq 0$}
        \If{$\gamma \notin S$}
          \State $S \gets S \cup \{\gamma\}$
          \State $\mathcal{P} \gets \mathcal{P} \cup \{\gamma\}$
        \EndIf
      \EndFor
    \EndWhile
    \State Let $S = \{\alpha,\beta^1,\dots,\beta^k\}$ be the final set of indices.
    \State Form the moment vector
           $m(t) \gets \big(\EE[X_t^{\alpha}],\EE[X_t^{\beta_1}],\dots,\EE[X_t^{\beta_k}]\big)^\top$.
    \State Use \cref{thm:dykin} to obtain the closed linear ODE system
           $\dfrac{\mathrm{d}}{\mathrm{d}t} m(t) = A\, m(t) +c$.
  \end{algorithmic}
\end{algorithm}

We formalize this process by constructing a set of monomials that is closed under the action of $\mathcal{A}$, as described in~\cref{alg:moment-closure} (the \emph{moment closure algorithm}). The key steps are summarized as follows:
\begin{enumerate}
  \item \textbf{Initialization.} Initialize the set of monomials \(S := \{\xx^\alpha\}\).
  \item \textbf{Closure construction.} While there exists a monomial $\xx^\beta \in S$ that has not yet been processed, compute $\mathcal{A}(\xx^\beta)$ and expand it as a linear combination of monomials (line 4-5 in \cref{alg:moment-closure}):
        \[
        \mathcal{A}\xx^\beta = \sum_{\gamma} a_{\beta\gamma}\,\xx^\gamma + c_\beta.
        \]
        For each new monomial $\xx^\gamma$ appearing on the right-hand side that is not already in $S$, add $\xx^\gamma$ to $S$. Mark $\xx^\beta$ as processed. (line 6-9 in \cref{alg:moment-closure})
  \item \textbf{Exact moment calculation.} If this procedure terminates after finitely many steps, we obtain a finite multi-index set $ S = \{\alpha, \beta_1,\dots,\beta_k\}$.  In this case, defining the moment vector
$
m(t) \defeq \big(\EE[X_t^{\alpha}],\EE[X_t^{\beta_1}],\dots,\EE[X_t^{\beta_k}]\big)^\top,
$
Dynkin’s formula yields a closed linear ODE system
\begin{equation}\label{eq:linear_system}
    \frac{\mathrm{d}}{\mathrm{d}t} m(t) = A\, m(t)\, + \, c\,,
\end{equation}
where the matrix \(A\) and the vector $c$ collects the coefficients from the
generator expansions.
\end{enumerate} 

If~\cref{alg:moment-closure} terminates, the resulting finite-dimensional ODE system yields explicit expressions for all moments in the set $S$, and in particular for the target moment $\EE[X_t^\alpha]$. Consequently, whenever the closure procedure terminates, the SDE is moment-solvable for $\alpha$. We formalize it as the following theorem.

\begin{theorem}[Explicit moment computation] \label{thm:moment_cal}
If~\cref{alg:moment-closure} terminates, then SDE~\eqref{eq:sde} is moment-solvable for $\alpha$. 
In particular, the moment vector $m(t)$ can be explicitly expressed as
{
\setlength{\abovedisplayskip}{3pt}
\setlength{\belowdisplayskip}{3pt}
\begin{equation}\label{eq:explicit_solution}
\big(\EE[X_t^{\alpha}],\EE[X_t^{\beta_1}],\dots,\EE[X_t^{\beta_k}]\big)^\top \eeq m(t) 
\eeq e^{A t} m(0)
  + \int_{0}^{t} e^{A (t-s)}\, c \, \dif s.
\end{equation}
}where $m(0) = \big(\EE[X_0^{\alpha}],\,\EE[X_0^{\beta_1}],\,\dots,\,\EE[X_0^{\beta_k}]\big)^\top$ is determined by the initial distribution, and the target moment $\EE[X_t^{\alpha}]$ is given by the first component of $m(t)$.
\end{theorem}
\begin{proof}
    Since \cref{alg:moment-closure} terminates, we obtain a linear ODE system \cref{eq:linear_system} for $m(t)$, whose solution is given explicitly by \cref{eq:explicit_solution}. This completes the proof. \qed
\end{proof}
    
\begin{remark}
    The matrix exponential and integral appearing on the right-hand side of~\eqref{eq:explicit_solution} can typically be evaluated in closed form using standard symbolic computation tools (e.g., \textsc{Mathematica}), making the explicit computation of $m(t)$ readily achievable in practice. We permit such closed-form representations to contain implicit symbolic expressions, as the evaluation of the matrix exponential generally entails computing eigenvalues, which may include implicit algebraic quantities.
\end{remark}

Depending on the specific moment $m_\alpha$ and the structure of the drift and diffusion coefficients, \cref{alg:moment-closure} may either terminate or diverge. The following example illustrates a nonlinear dynamics for which the closure procedure terminates, resulting in a finite-dimensional moment system.

\begin{example}[Chemical process in an Ornstein–Uhlenbeck environment~\cite{kallianpur1994stochastic}]\label{ex:running_example}
Consider a chemical system, where $X_t$ models a fluctuating environment and $Y_t$ denotes the concentration of a chemical species influenced by the environment:
{
\setlength{\abovedisplayskip}{4pt}
\setlength{\belowdisplayskip}{4pt}
\begin{equation}
\label{eq:env-chem-sde}
\begin{cases}
\dif X_t
= - X_t \dif t +  \dif W_t^{(1)}, \\[0.4em]
\dif Y_t
= \bigl(-2 Y_t + X_t +  X_t^{2} \bigr)\dif t
  +  X_t \dif W_t^{(2)},
\end{cases}
\end{equation}
}Suppose the initial state is $(X_0, Y_0) = (0, 0)$. We seek to compute the second moment of $Y_t$, i.e., $m_{(0,2)} = \EE[Y_t^2]$. Applying~\cref{alg:moment-closure}, we obtain a closed $8$-dimensional linear ODE system for the collection of moments
{
\setlength{\abovedisplayskip}{5pt}
\setlength{\belowdisplayskip}{5pt}
\[ 
m_{(0,2)},\quad m_{(2,1)},\quad m_{(2,0)},\quad m_{(1,1)},\quad m_{(4,0)},\quad m_{(3,0)},\quad m_{(0,1)},\quad m_{(1,0)} 
\]
}where $m_{(i,j)} \defeq \EE[X_t^i Y_t^j]$. Solving this system (see Appendix~\ref{appendix:moment-ode} for the explicit ODE system), we obtain an explicit formula for the second moment:
\begin{equation}
\mathbb{E}\bigl[Y_t^2\bigr]
= \frac{1}{3}
+ \frac{2}{3} e^{-3t}
+ \left(-\frac{t}{4} - \frac{11}{8}\right)e^{-2t}
+ \left(\frac{3}{4} t^{2} + t + \frac{3}{8}\right)e^{-4t}\,.
\tag*{\qedT}
\end{equation}
\end{example}



The following example illustrates a case in which \cref{alg:moment-closure} does \emph{not} terminate.

\begin{example}[Double-well potential~\cite{gardiner2004handbook}]
Consider the bistable Langevin system describing the dynamics of a particle in a double-well potential,
\[
\mathrm dX_t = \bigl(a X_t - X_t^{3}\bigr)\,dt + \sigma\,dW_t,
\]
where $a, \sigma \in \mathbb{R}$ are constants.  Applying the generator $\mathcal{A}$ to the monomial $x^n$ for $n \geq 2$, we obtain $\mathcal{A} x^n = - n x^{n+2} + a n x^n + \frac{\sigma^2}{2} n(n-1) x^{n-2}$.
Therefore, starting from $x^n$, each application of the generator introduces a new, higher-degree monomial $x^{n+2}$, and recursively, all monomials of the form $x^{n+2k}$ for $k \geq 0$ are generated. As a result, the set $S$ continues to expand indefinitely, and~\cref{alg:moment-closure} does not terminate.
\qedT
\end{example}

Combining~\cref{thm:moment_cal} with~\cref{def:moment-solvable}, we obtain the following characterization: 

\begin{theorem}\label{thm:termination_implies_solvable}
    Given an SDE~\eqref{eq:sde}, if~\cref{alg:moment-closure} terminates for every monomial $\xx^\alpha$, then the SDE~\eqref{eq:sde} is moment-solvable; that is, explicit closed-form expressions can be computed for all moments.
\end{theorem}


\section{A Class of SDEs with Moment-Solvable Property}\label{sec:solvable_SDEs}
In this section, we identify a class of SDEs, termed \emph{pro-solvable SDEs}, for which \cref{alg:moment-closure} terminates for any monomial $\xx^\alpha$, thereby ensuring the moment-solvable property by \cref{thm:termination_implies_solvable}. The class of pro-solvable SDEs contains both linear SDEs and certain nonlinear SDEs whose variables exhibit a triangular dependence structure. We present the definition of pro-solvable SDEs in \cref{subsec:pro-solvable} and prove \cref{alg:moment-closure} terminates for any pro-solvable SDE and any multi-index $\alpha$ in \cref{subsec:termination-pro-solvable}.

\subsection{Pro-Solvable SDEs}\label{subsec:pro-solvable}

We begin by formalizing the concept of \emph{ordered partition}, which serves as a foundation for introducing the notion of pro-solvable SDEs.

\begin{definition}[Ordered partition]
We say that non-empty blocks $G_1, G_2, \dots, \\G_r$ form an \emph{ordered partition} of $\{1,2,\dots, n\}$ if 
\[
\{1,2,\dots,n\} = G_1 \cup \cdots \cup G_r, \qquad G_i \cap G_j = \emptyset \quad \text{for any } i\neq j,
\]
and the blocks $G_1, G_2, \dots, G_r$ are equipped with the natural order $G_1 \prec \cdots \prec G_r$. For each $p \in \{1, \dots, r\}$, let $\xx^{(p)} \defeq (x_i)_{i\in G_p}$ denote the collection of variables whose indices belong to block $G_p$, and $\xx^{(<p)} \defeq (x_i)_{i \in G_1 \cup \cdots \cup G_{p-1}}$ be the collection of all variables whose indices belong to the preceding blocks of $G_p$.
\end{definition}

\begin{definition}[Pro-solvable SDEs]\label{def:pro-solvable} 
An SDE~\eqref{eq:sde} is \emph{pro-solvable} if there exists an ordered partition $G_1, G_2, \dots, G_r$ of $\{1,2,\dots, n\}$, such that its drift $b = (b_1, \dots, b_n)$ and diffusion matrix $\sigma = (\sigma_{ik})_{1 \leq i \leq n,\, 1 \leq k \leq m}$ satisfy the \emph{block-triangular affine} structure: for every block $G_p$, and for all $i \in G_p$, $k = 1, \dots, m$,
\begin{align}
b_i(\xx) &= \sum_{j \in G_p} \Lambda_{ij} x_j + B_i\bigl(\xx^{(<p)}\bigr), 
&&\Lambda_{ij} \in \RR, \;\; B_i \in \RR[\xx^{(<p)}], 
\label{eq:block-b}\\
\sigma_{ik}(\xx) &= \sum_{j \in G_p} A_{ikj} x_j + P_{ik}\bigl(\xx^{(<p)}\bigr), 
&&A_{ikj} \in \RR, \;\; P_{ik} \in \RR[\xx^{(<p)}].
\label{eq:block-sigma}
\end{align}
In particular, each $b_i$ and $\sigma_{ik}$ is affine-linear in the variables of its own block $\xx^{(p)}$, and any nonlinearity depends only on variables from earlier blocks $\xx^{(<p)}$.
\end{definition}

The block-triangular affine structure in pro-solvable SDEs imposes a natural hierarchy among the variables: within each block $G_p$, the drift and diffusion coefficients are affine-linear functions of the variables in that block, while any nonlinear dependence is restricted to variables in preceding blocks $\xx^{(<p)}$. It ensures that the moment dynamics associated with variables in higher-indexed blocks only depend on moments of lower-indexed blocks, and never vice versa.

\begin{example}\label{ex:running_example2}
Reconsider the SDE in \cref{ex:running_example}, given in \cref{eq:env-chem-sde}. Suppose its drift and diffusion coefficients are
\[
b(x_1,x_2) =
\begin{pmatrix}
-\,x_1\\[0.2em]
-2x_2 + x_1 + x_1^2
\end{pmatrix},
\qquad
\sigma(x_1,x_2) =
\begin{pmatrix}
1 & 0\\
0 & x_1
\end{pmatrix}.
\]
It is straightforward to verify that SDE~\eqref{eq:env-chem-sde} is pro-solvable under the ordered partition $G_1 = \{1\}$, $G_2 = \{2\}$ with $r = 2$. \qedT
\end{example}

The class of pro-solvable SDEs subsumes several important subclasses. In particular, when the ordered partition is taken as (i) $r = 1$ with $G_1 = \{1,2,\dots, n\}$, the pro-solvable SDEs specialize to the well-known class of linear SDEs, in which both the drift and diffusion coefficients are affine functions of all variables. On the other hand, when the partition is (ii) $r = n$ with $G_i = \{i\}$ for $1 \leq i \leq n$, the pro-solvable condition reduces to the strictly triangular case, where each variable may depend nonlinearly only on those variables with strictly smaller indices. 


\begin{remark}
    A trivial approach of checking pro-solvability of SDEs by enumerating all possible ordered partitions and then checking block triangular affine structure, is exponential in system dimension.
    A more efficient polynomial-time alternative, similar to the solvability checking of the recurrence relation in \cite[Sect. 4]{amrollahi2025solvable}, by constructing a dependency graph of variables, an SDE is pro-solvable iff no strongly connected component contains a nonlinear edge.
\end{remark}

\subsection{Pro-solvable SDEs are moment-solvable} \label{subsec:termination-pro-solvable}
In this subsection, we establish that pro-solvable SDEs guarantee termination of the iterative moment closure procedure described in \cref{alg:moment-closure} (Lines 3–9). That is, for any monomial $\xx^\alpha$, \cref{alg:moment-closure} generates only finitely many new moments, thereby ensuring the moment-solvable property. 

Clearly, the termination of \cref{alg:moment-closure} hinges on whether infinitely many new monomials are added to the set $\mathcal{P}$. The intuitive strategy is to construct a ranking function over monomials such that, whenever a new monomial $\xx^\gamma$ arises from the generator action $\mathcal{A}\xx^\beta$, its rank is no greater than that of $\xx^\beta$.

To formalize this idea, we examine in detail how the operator $\mathcal{A}$ acts on pro-solvable SDEs. For such systems, the generator $\mathcal{A}$ expands as
{\small
\begin{align}
 \mathcal{A}\,f
&   \eeq  \sum_{i=1}^n b_i(\xx)\partial_{x_i} f
+ \frac{1}{2}\sum_{i,j=1}^n \sum_{k=1}^m \sigma_{ik}(\xx)\sigma_{jk}(\xx) \,
  \partial_{x_i}\partial_{x_j} f  \notag \\
 &   \eeq  \sum_{i=1}^n \left(\sum_{l \in G_{I(i)}} \Lambda_{il} x_l + B_i\bigl(\xx^{(<I(i))}\bigr)\right)\partial_{x_i} f  \label{eq:unfold_prosolvable} \\
 + \frac{1}{2} &\sum_{i,j=1}^n   \sum_{k=1}^m   \left(\sum_{l \in G_{I(i)}} A_{ikl} x_l + P_{ik}\bigl(\xx^{(<I(i))}\bigr)\right)\left(\sum_{l \in G_{I(j)}} A_{jkl} x_l + P_{jk}\bigl(\xx^{(<I(j))}\bigr)\right) \,
  \partial_{x_i}\partial_{x_j} f \notag
\end{align}}

\noindent where $I(i)$ denotes the unique block index for which $i$ belongs to $G_{I(i)}$, and $\partial_{x_i}$ is the partial differential operator that maps a function $f$ to $\nicefrac{\partial f}{\partial x_i}$. Thus, the operator $\mathcal{A}$ is a linear combination of primitive terms of the form
\begin{equation}\label{eq:primitive_term}
\xx^\gamma \,\partial_{x_i},\qquad
\xx^\gamma \,\partial_{x_i}^2,\qquad
\xx^\gamma \,\partial_{x_i}\partial_{x_j},    
\end{equation}
where the monomial $\xx^\gamma$ arises from the polynomial coefficients of $b_i$ and $\sigma_{ij}$. Depending on how $\xx^\gamma$ is produced, we distinguish two types of primitive terms.

\begin{definition}[Classification of primitive terms in $\mathcal{A}$]\label{def:classification}
Suppose SDE~\eqref{eq:sde} is pro-solvable with ordered partition $G_1, G_2, \dots, G_r$, then the primitive terms in operator $\mathcal{A}$ are classified into
\begin{itemize}
  \item[(i)] \emph{Linear-produced terms}: if monomial $\xx^\gamma$ in primitive term comes entirely from the affine linear parts
  $
    \sum_{l\in G_p} \Lambda_{il}x_l,
  $
  $
    \sum_{l\in G_p} A_{ikl} x_l
  $ or $\sum_{l\in G_q} A_{jkl} x_l$
  with no factor from any $B_i$, $P_{ik}$ or $P_{jk}$, where $I(i)=p, I(j) = q$. A monomial in linear-produced terms is not necessarily linear, as we allow the product of two affine linear parts.

  \item[(ii)] \emph{Polynomial-produced terms}: if monomial $\xx^\gamma$ in primitive term contains at least one factor
  from some $B_i(\xx^{(<I(i))})$, $P_{ik}(\xx^{(<I(i))})$, or $P_{jk}(\xx^{(<I(j))})$.
\end{itemize}
\end{definition}


For a polynomial-produced primitive term, we define its \emph{source block index} as follows. If the monomial $\xx^\gamma$ contains a factor from $B_i$ or $P_{ik}$ for some $i \in G_p$, then $p$ is regarded as a candidate source block index. In cases where multiple candidate source block indices exist (e.g., cross terms involving both $P_{ik}$ and $P_{jk}$), we select the source block index as the larger one (e.g. $\max\{p,q\}$ if $i \in G_p$ and $j\in G_q$ in the cross term case). Consequently, every polynomial-produced primitive term is associated with a unique source block index $p \in \{1,\dots,r\}$.

\begin{example}
Continuing \cref{ex:running_example2}, the generator $\mathcal{A}$ corresponding to SDE~\eqref{eq:env-chem-sde} simplifies to
\begin{align}
\label{eq:instantiated-generator}
\mathcal{A}
 \eeq -x_1\,\partial_{x_1} 
  + \bigl(-2x_2 + x_1 + x_1^2\bigr)\,\partial_{x_2} 
  + \frac{1}{2}\,\partial_{x_1}^2 
  + \frac{1}{2} x_1^2\,\partial_{x_2}^2 .
\end{align}
Following the classification in \cref{def:classification}, the linear-produced terms $\{x_1\,\partial_{x_1}, x_2\,\partial_{x_2}\}$ are obtained, and the polynomial-produced terms are $\{ x_1\,\partial_{x_2}, x_1^2\,\partial_{x_2}, \partial^2_{x_1}, x_1^2\partial^2_{x_2}\}$. Moreover, the corresponding source block indices for the polynomial-produced terms $x_1\,\partial_{x_2}$, $x_1^2\,\partial_{x_2}$, $\partial^2_{x_1}$, and $x_1^2\,\partial^2_{x_2}$ are $2$, $2$, $1$, and $2$ respectively. \qedT
\end{example}


We proceed to analyze how the exponents change when applying $\mathcal{A}$ to a monomial $\xx^\beta$. To facilitate this analysis, we introduce the notion of the \emph{block exponent sum} and \emph{block exponent difference}.

\begin{definition}
Given the notations above, let $\beta$ be a multi-index with $\xx^\beta = \prod_{i=1}^n x_i^{\beta_i}$. For each block $G_p$, the block exponent sum is defined as
\[
s_p(\beta) \;\defeq\; \sum_{i\in G_p} \beta_i, \qquad\text{for } p = 1, \dots, r.
\]
Moreover, for any two multi-indices $\beta$ and $\beta'$, the block exponent difference is defined by $\Delta s_p(\beta, \beta') \defeq s_p(\beta') - s_p(\beta)$ for $1 \leq p \leq r$.
\end{definition}

Now, given a monomial $\xx^\beta$, let $\xx^{\beta'}$ be any monomial that appears in the expansion of $\mathcal{A} \xx^\beta$. According to \cref{eq:unfold_prosolvable}, there must exist a unique primitive term $T$ of the form $\xx^\gamma \,\partial_{x_i}$, $\xx^\gamma \,\partial_{x_i}^2$, or $\xx^\gamma \,\partial_{x_i}\partial_{x_j}$ such that $T \xx^\beta = c\,\xx^{\beta'}$ for some constant $c$. 
The following result shows that the change in exponents is bounded, and this bound is independent of the specific monomials $\xx^\beta$ and $\xx^{\beta'}$.

\begin{lemma}[Bound on exponent change]\label{lem:bound_on_exponent_change}
Given a monomial $\xx^\beta$, let $\xx^{\beta'}$ be any monomial that appears in the expansion of $\mathcal{A} \xx^\beta$. Suppose $\xx^{\beta'}$ is produced by a primitive term $T$ (i.e., $T$ is of the form given in \cref{eq:primitive_term}, and $T \xx^\beta = c\,\xx^{\beta'}$ for some constant $c$). Then the following properties hold:
\begin{itemize}
    \item If $T$ is a linear-produced term, then $\Delta s_q(\beta, \beta') \le 0$ for all $q = 1, \dots, r$.
    \item If $T$ is a polynomial-produced term with source block index $p$, then 
    \begin{enumerate}
        \item[(1)] $\Delta s_p(\beta, \beta') \le -1$;
        \item[(2)] $\Delta s_q(\beta, \beta') = 0$ for all $q > p$;
        \item[(3)] For each $q < p$, there exists a constant $C_{p, q} \in \mathbb{N}$ (independent of $\beta$ and $\beta'$, depending only on $q$ and on the degrees of $B_i$ and $P_{ik}$ for all $i$ in $G_1 \cup \cdots \cup G_{p}$) such that $0 \le \Delta s_q(\beta, \beta') \le C_{p, q}$.
    \end{enumerate}
\end{itemize}
\end{lemma}


\begin{proof}
The proof proceeds by a case-by-case analysis of the primitive term $T$.

\emph{Case 1: $T$ is a linear-produced term.}
Suppose $T$ is a linear-produced term, then $T$ can only take the form in either $x_l \,\partial_{x_i}$ for some $i, l\in G_p$ or $x_{l} x_{l'}\,\partial_{x_i}\partial_{x_j}$ for some $i, l\in G_p$ and $j, l' \in G_{p'}$. 
A direct calculation shows that $\Delta s_q \le 0$ for all $q=1,\dots,r.$
Indeed, each derivative $\partial_{x_i}$ reduces the exponent of some $x_i$ in its
block by $1$, and the linear coefficient can reintroduce at most as many variables in
that block as there are derivatives; hence, the total exponent in each block never
increases under linear-produced terms.

\emph{Case 2: $T$ is a polynomial-produced term.} Suppose $T$ is a polynomial-produced term with source block index $p$. By \cref{eq:unfold_prosolvable} and the definition of the source block index, $T$ always contains at least one derivative with respect to some $x_i$ in block $G_p$, and $T$ can only take the form in the following three cases: 
\begin{itemize}
  \item[$\ast$] $\xx^\gamma \partial_{x_i}$ with $\xx^\gamma$ a monomial in $B_i(\xx^{(<p)})$;
  \item[$\ast$] $\xx^\gamma \partial_{x_i}^2$ with $\xx^\gamma$ a monomial in
        $2 x_l P_{ik}(\xx^{(<p)})$ or $\left(P_{ik}(\xx^{(<p)})\right)^2$ for some $k$ and some $l \in G_p$;
  \item[$\ast$] $\xx^\gamma \partial_{x_i}\partial_{x_j}$ with $\xx^\gamma$ a monomial in $P_{ik}(\xx^{(<p)}) P_{jk}(\xx^{(<u)})$ for some $k$, and for some $j \neq i$ with $j \in G_u$ and $u \leq p$.
\end{itemize}
In either case, we can directly check that:
\begin{enumerate}
  \item[(1)] $\Delta s_p \le -1$, i.e.\ the total exponent in the source block $G_p$
        strictly decreases;
  \item[(2)] $\Delta s_q = 0$ for all $q>p$, i.e.\ no later block is affected;
  \item[(3)] For each $q<p$, there exists a constant $C_{p,q}\in\Nats$ (\emph{depending only on} $q$ and on the degrees of $B_i$ and $P_{ik}$ for all $i$ in $G_1 \cup \cdots \cup G_{p}$) such that
        $ 0 \le \Delta s_q \le C_{p,q}$.
\end{enumerate}
Intuitively, the above result implies each application of a polynomial-produced term from block $p$ differentiates at least once in some variable of $G_p$, thus reducing $s_p$ by one or two, while the coefficient can reintroduce at most one variable from $G_p$ and some bounded amount of variables from earlier blocks $G_1,\dots,G_{p-1}$. 
Combining both cases, we obtain the desired result. \qed
\end{proof}

At the beginning of this subsection, we mention that our strategy is to construct a ranking function over monomials to prove the termination of \cref{alg:moment-closure}. Now, we introduce the following weighted block degree serving as a ranking function, such that its value over a newly added monomial does not increase. 

\begin{definition}\label{def:weighted_block_degree}
Given the notations above, the \emph{weighted block degree} for a monomial is defined as 
\[
  \deg_W(\xx^\beta) \;\defeq\; \sum_{p=1}^r W_p\, s_p(\beta), 
\]
where the block weights $W_1,\dots,W_r$ are chosen as follows: set $W_1 \defeq 1$, and for each $p=2,\dots,r$, choose $W_p$ inductively by
\[
  W_p \;>\; \sum_{q<p} C_{p,q}\, W_q.
\]
\end{definition}

The following lemma shows that the weighted block degree is indeed non-increasing. 
Intuitively, by \cref{lem:bound_on_exponent_change}, no new higher-degree monomials are created within the same block. Any additional complexity can only arise from dependencies on variables in earlier blocks, whose closure is handled inductively. This behavior is precisely captured by the choice of block weights in \cref{def:weighted_block_degree}.

\begin{lemma}[Non-increase of weighted block degree]\label{lem:non-increase}
Given a monomial $\xx^\beta$, let $\xx^{\beta'}$ be any monomial that appears in the expansion of $\mathcal{A} \xx^\beta$, then $$\deg_W(\xx^{\beta'}) - \deg_W(\xx^\beta)
\lleq 0\,.$$
\end{lemma}
\begin{proof}
Given the notations above, suppose $T$ is a polynomial-produced term with source block $p$, then by \cref{lem:bound_on_exponent_change}
{
\setlength{\abovedisplayskip}{3pt}
\setlength{\belowdisplayskip}{3pt}
\[
  \deg_W(\xx^{\beta'}) - \deg_W(\xx^\beta)
  = \sum_{q=1}^r W_q\,\Delta s_q
  \;\le\;
  W_p\cdot(-1) + \sum_{q<p} W_q C_{p,q}
  \;<\; 0.
\]
}Therefore, any genuinely new monomial produced by a polynomial-produced term strictly
\emph{decreases} the weighted block degree. 
On the other hand, if $T$ is a linear-produced term,  we have, by \cref{lem:bound_on_exponent_change}
{
\setlength{\abovedisplayskip}{3pt}
\setlength{\belowdisplayskip}{3pt}
\[
  \deg_W(\xx^{\beta'}) - \deg_W(\xx^\beta)
  = \sum_{q=1}^r W_q\,\Delta s_q
  \;\le\; 0,
\]
}Hence, the weighted degree is \emph{nonincreasing} under linear-produced terms. This completes the proof.\qed
\end{proof}

Consider the directed graph whose vertices are monomials and with an edge
$\xx^\beta \to \xx^{\beta'}$ whenever $\xx^{\beta'}$ appears in
$\mathcal{A} \xx^\beta$. Starting from $\xx^\alpha$, any monomial that appears in some
$\mathcal{A}^m(x^\alpha)$ is reachable via a directed path.
Along any path of distinct vertices $
\xx^{\beta^{(0)}} = \xx^\alpha \to \xx^{\beta^{(1)}} \to \xx^{\beta^{(2)}} \to \cdots,
$
the sequence $\deg_W(\xx^{\beta^{(k)}})$ is non-increasing by \cref{lem:non-increase}. Thus, for every reachable monomial $\xx^\beta$ we have
\[
  \deg_W(\xx^\beta) \;\le\; \deg_W(\xx^\alpha).
\]
Since all weights $W_p$ are positive, it implies that each block sum
$s_p(\beta)$ is bounded:
\[
  0 \le s_p(\beta) \le \frac{\deg_W(\xx^\alpha)}{W_p},\qquad p=1,\dots,r.
\]
In particular, each individual exponent is bounded:
\[
  0 \le \beta_i \le s_p(\beta) \le \frac{\deg_W(\xx^\alpha)}{W_p}
  \quad\text{for all } i\in G_p,\;p=1,\dots,r.
\]
Therefore, the set of all multi-indices $\beta\in\Nats^n$ with
$\deg_W(\xx^\beta)\le \deg_W(\xx^\alpha)$ is finite, hence only finitely many monomials are
reachable from $\xx^\alpha$. 
Based on the above analysis, we present the main result for pro-solvable SDEs as follows.
\begin{theorem}\label{thm:pro-solvable}
    If SDE \eqref{eq:sde} is pro-solvable, then \cref{alg:moment-closure} terminates after finitely many iterations for any monomial $\xx^\alpha$, and SDE \eqref{eq:sde} is moment-solvable.
\end{theorem}
\begin{proof}
By \cref{lem:non-increase}, only finitely many monomials can appear when iteratively executing lines 3--9 of \cref{alg:moment-closure}, starting from any $\xx^\alpha$. It implies that \cref{alg:moment-closure} terminates for any multi-index $\alpha$, hence SDE~\eqref{eq:sde} is moment-solvable by \cref{thm:termination_implies_solvable}. \qed
\end{proof}


\myparagraph{Complexity Analysis}
The computational complexity of our method for exact moment estimation arises primarily from two sources: constructing the moment closure set $S$ in \cref{alg:moment-closure}, and solving the resulting linear ODE system~\eqref{eq:linear_system}.

\smallskip

\emph{(1) Complexity of closure construction.}  
For a fixed pro-solvable SDE and initial multi-index $\alpha$, let $S_\alpha$ denote the finite set of monomials reachable from $x^\alpha$ by \cref{alg:moment-closure}. By \cref{lem:non-increase} and the argument preceding \cref{thm:pro-solvable}, there exists a constant $C_0$ (depending only on the SDE) such that for every $x^\beta \in S_\alpha$,
\[
|\beta| \eeq \sum_{i=1}^n \beta_i \lleq \sum_{p=1}^r\frac{\deg_W(\xx^\alpha)}{W_p} \lleq C_0 |\alpha|\,.
\]
Hence, $S_\alpha$ is contained within the set of all monomials of total degree at most $C_0|\alpha|$ in $n$ variables, which gives the combinatorial bound
\[
  |S_\alpha| \leq \binom{n + C_0|\alpha|}{C_0|\alpha|}.
\]
In particular, for fixed dimension $n$, this yields $|S_\alpha| = O\bigl((|\alpha|)^n\bigr)$, while for fixed moment order $k = |\alpha|$, we have $|S_\alpha| = O\bigl(n^k\bigr)$. Since \cref{alg:moment-closure} processes each element of $S_\alpha$ at most once, the time and memory complexity of closure construction (cf.\ lines 3 -- 9 in \cref{alg:moment-closure}) is $O(|S_\alpha|)$. 
The hidden constants in the big-O notation depend only on the given SDE, in particular on the degrees and number of monomials in the polynomial drift and diffusion coefficients.

\smallskip

\emph{(2) Complexity of solving the ODE system.}  
According to \cref{eq:explicit_solution} in \cref{thm:moment_cal}, solving the resulting ODE system reduces to computing the matrix exponential $e^{At}$ where $A$ is a $|S_\alpha|$ dimensional matrix.
Using standard algebraic methods (e.g., Jordan or rational canonical form), this can be done in time polynomial in $|S_\alpha|$, with worst-case complexity $\mathcal{O}(|S_\alpha|^3)$. \qedT

\section{Experiments}\label{sec:experiments}
To demonstrate the effectiveness and applicability of our EME framework, we implemented~\footnote{Available at \url{https://github.com/Shenghua-Feng/Exact_Moment_Estimation}} \cref{alg:moment-closure} in Python~3.13, leveraging standard symbolic and numerical linear algebra libraries. Given a polynomial SDE~\eqref{eq:sde} and a target monomial $x^\alpha$, our prototype automatically constructs the moment-closure set $S_\alpha$, derives the corresponding linear ODE system, and computes the desired $\alpha$-moment. 

\myparagraph{Benchmarks}
We evaluated our method on a suite of SDE benchmarks (details see Appendix~\ref{appendix::benchmarks}), encompassing both linear and nonlinear examples in the literature as well as models with practical relevance. Specifically, we present in detail two cases: a consensus network with noise~\cite{olfati2007consensus} and a nonlinear vehicle platoon system adapted from~\cite{kavathekar2011vehicle} to illustrate our method and demonstrate its usefulness for verification problems. All experiments were performed on a MacBook Pro with an Apple~M4 processor, 16~GB of RAM, and running macOS Sequoia.

\subsection{Case studies} \label{subsec:case_study}

\myparagraph{Consensus network with noise~\cite{olfati2007consensus}} This model describes an $n+1$ dimensional multi-agent consensus network (e.g., distributed sensors or robots) with noisy communication:
\[
\mathrm{d}X_i(t)
=
\Bigl(
  a_i - \lambda X_i(t)
  + \kappa\bigl(X_{i+1}(t) - 2 X_i(t) + X_{i-1}(t)\bigr)
\Bigr)\,\mathrm{d}t
+ \sigma_i X_i(t)\,\mathrm{d}W_t^{(i)},
\]
for $i = 1,\dots,n+1$, with periodic boundary conditions $
X_0(t) = X_n(t), X_{n+1}(t) = X_1(t)$. Since this dynamic is linear, it is pro-solvable by \cref{def:pro-solvable}. Consider the case $n=2$ with parameter values $a_1 = 0$, $a_2 = 0$, $\lambda = 1$, $\kappa = 0.5$, $\sigma_1 = 1$, $\sigma_2 = 1$, and initial state $(1,0)$. The verification objective is to ensure that, with probability at least $1 - e^{-t}$, the disagreement between the two agents, encoded by $|X_1(t) - X_2(t)|$, remains less than $0.1$ for any $t \geq 10$, that is,
\[P\left( \left|X_1(t) - X_2(t)\right|\geq 0.1 \right) \leq e^{-t} \quad \text{for }t\geq 10. \]
To this end, we compute $\EE[(X_1(t) - X_2(t))^2]$. By explicitly calculating the moments $\EE[X_1(t)^2]$, $\EE[X_1(t)X_2(t)]$, and $\EE[X_2^2(t)]$, we obtain
\[ \EE[(X_1(t)-X_2(t))^2] \eeq  \frac{(17-3\sqrt{17})e^{\frac{(\sqrt{17}-7)t}{2}}
      + (17+3\sqrt{17})e^{-\frac{(\sqrt{17}+7)t}{2}}}{34}. \]    
Consequently, by applying Markov's inequality, we obtain, for any $t \geq 10$,
$$ P\left(|X_1(t) - X_2(t)| \geq 0.1 \right) \leq \frac{\EE[(X_1(t)-X_2(t))^2]}{0.1^2} \leq e^{-t}\;. $$ 
This verifies the goal. It is worth noting that, to the best of our knowledge, this verification problem cannot be solved directly by a martingale-based approach that seeks a polynomial $h(t, x_1, x_2)$ satisfying the supermartingale condition to upper bound $\EE[(X_1(t)-X_2(t))^2]$. Specifically, if such a polynomial $h$ existed, we have $P(|X_1(t) - X_2(t)| \geq 0.1) \leq 100\, \EE[(X_1(t)-X_2(t))^2] \leq 100\, \EE[h(t,X_1(t), X_2(t))] \leq 100\, h(0, x_1(0), x_2(0)).$ Since there must exist some $T>10$ such that $e^{-T} < 100\, h(0, x_1(0), x_2(0))$, the standard martingale-based method cannot certify the property holds for all $t\geq 10$. \qedT



\medskip

\myparagraph{Vehicle platoon~\cite{kavathekar2011vehicle}}
We consider a nonlinear system adapted from~\cite{kavathekar2011vehicle} that involves two vehicles moving along a straight lane. For each vehicle, denote its position and velocity by $(p_1(t), v_1(t))$ and $(p_2(t), v_2(t))$, respectively. The first vehicle acts as the leader and follows a stochastic acceleration model:
\[
\mathrm{d}p_1 = v_1\,\mathrm{d}t, \qquad 
\mathrm{d}v_1 = \bigl(-a_1 v_1 + u_1\bigr)\,\mathrm{d}t
               + \sigma_1\,\mathrm{d}W_t^{(1)},
\]
where $a_1 > 0$ is a damping coefficient, $u_1$ is a control input (desired acceleration), and $\sigma_1$ scales the driving noise $W^{(1)}_t$. The second vehicle implements a nonlinear control law based on both its own velocity and the velocity of the leader:
\[
\begin{aligned}
\mathrm{d}p_2  = v_2\,\mathrm{d}t, \qquad
\mathrm{d}v_2  = \left(- a_2 v_2 + (v_1 - 1)^2 \right)\,\mathrm{d}t 
   + \sigma_2\,\mathrm{d}W_t^{(2)},
\end{aligned}
\]
where $a_2$ and $\sigma_2$ are parameters. It can be checked that this system is pro-solvable under the ordered partition $G_1 = \{v_1\}$, $G_2 = \{p_1, p_2, v_2\}$. Consider the parameter instantiation $a_1 = 1$, $u_1 = 1$, $\sigma_1 = 1$, $a_2 = 1$, and $\sigma_2 = 1$, with initial state $(p_1, v_1, p_2, v_2) = (1, 0, 0, 0)$. Suppose that the verification objective is to ensure that the expected distance between the vehicles, $\EE[p_1(t) - p_2(t)]$, always remains between $0.5$ and $1.5$ for all $t\geq 0$.

To verify this, we compute the expected values $\EE[p_1(t)]$ and $\EE[p_2(t)]$, yielding
\[
\frac{3}{4} \leq \EE\bigl[p_1(t) - p_2(t)\bigr]
= \frac{3}{4} + \frac{e^{-t}}{2} - \frac{e^{-2t}}{4} \leq 1,
\]
which verifies the desired safety property.\qedT

\setlength{\tabcolsep}{1.1mm}
\begin{table}[!t]
\captionsetup{font={small}}
\caption{Experimental results for exact moment calculation.}
\label{tab:benchmarks}
\begin{center}
  \begin{tabular}{l ccc c cc c crc c cr}
    \toprule
    ~ &\multicolumn{3}{c}{SDE System} & ~ & \multicolumn{2}{c}{Moment} &~ & \multicolumn{3}{c}{Obtained Closure $S_\alpha$}& ~ & \multicolumn{2}{c}{Solve ODE}\\ 
    \cmidrule{2-4} \cmidrule{6-7} \cmidrule{9-11} \cmidrule{13-14}
    Benchmark & dim & deg & p-s & & $\xx^\alpha$ & $|\alpha|$ & & succ & time & $|S_\alpha|$ & & succ & time\\
    \midrule
    
    \multirow{5}{*}{\textsf{ou-env}~\cite{kallianpur1994stochastic}}  & \multirow{5}{*}{2} & \multirow{5}{*}{2} & \multirow{5}{*}{yes} && $\EE[x_2^2]$ & 2   && \Checkmark & 0.01s & 8 && \Checkmark & 0.2s\\
    ~ & ~ & ~ & ~ && $\EE[x_2^3]$ & 3   && \Checkmark & 0.02s & 15 && \Checkmark & 0.6s\\
    ~ & ~ & ~ & ~ && $\EE[x_2^4]$ & 4   && \Checkmark & 0.02s & 24 && \Checkmark & 1.2s\\
    ~ & ~ & ~ & ~ && $\EE[x_2^5]$ & 5   && \Checkmark & 0.04s & 35 && \Checkmark & 2.8s\\
    ~ & ~ & ~ & ~ && $\EE[x_2^{10}]$ & 10   && \Checkmark & 0.17s & 120 && \Checkmark & 39.4s\\
    \hline
    \multirow{3}{*}{\textsf{gene}~\cite{sinigh2007stochastic}} & \multirow{3}{*}{5} & \multirow{3}{*}{3} & \multirow{3}{*}{yes} && $\EE[x_1x_5]$ & 2 && \Checkmark & 0.04s & 23 && \Checkmark & 3.0s\\
    ~ & ~ & ~ & ~ && $\EE[x_5^2]$ & 2 && \Checkmark & 0.14s & 85 && \Checkmark & 79.6s\\
    ~ & ~ & ~ & ~ && $\EE[x_1x_5^2]$ & 3 && \Checkmark & 0.17s & 115 && \Checkmark & 164.1s\\
    \hline
    \textsf{consensus}~\cite{olfati2007consensus} & 2 & 1 & yes && $\EE[x_1x_2]$  & 2   && \Checkmark & 0.01s & 3 && \Checkmark & 0.2s\\
    \textsf{vehicles}~\cite{kavathekar2011vehicle} & 4 & 2 & yes && $\EE[x_2^2]$ & 2 && \Checkmark & 0.01s & 13 && \Checkmark & 0.5s\\
    \textsf{oscillator}~\cite{hafstein2018lyapunov} & 3 & 2 & yes && $\EE[x_2x_3^2]$ & 3 && \Checkmark & 0.01s & 6 && \Checkmark & 2.3s\\
    \textsf{coupled3d} & 3 & 3 & no && $\EE[x_1^2x_2^2]$ & 4 && \Checkmark & 0.01s & 3 && \Checkmark & 0.2s \\
    \bottomrule
  \end{tabular}
\end{center}
\scriptsize{
\textbf{dim}: Dimension of the SDE system; 
\textbf{deg}: Maximum polynomial degree of drift/diffusion terms in the SDE; 
\textbf{p-s}: Whether the SDE is pro-solvable; 
\textbf{$\xx^\alpha$}: Target moment to compute; 
\textbf{$|\alpha|$}: Degree of the target moment; 
\textbf{succ}: Whether a closed linear ODE system was successfully constructed and solved; 
\textbf{time}: Time required to obtain the closure $S_\alpha$ (i.e., construct or solve the linear ODE system); 
$\mathbf{|S_\alpha|}$: Dimension of the resulting linear ODE system. 
}
\end{table}

\subsection{Evaluation of effectiveness}
Table~\ref{tab:benchmarks} summarizes the experimental results of our method on a diverse suite of polynomial SDE benchmarks,  
which cover a range of system dimensions, polynomial degrees, and moment orders, illustrating the generality of our approach.

\myparagraph{Efficiency and scalability}
For all pro-solvable SDEs, 
our method successfully constructs the finite closures and computes the exact moment for all tested cases.
Closure construction times are consistently short, and the dimension $|S_\alpha|$ scales polynomially with the moment order and system size (cf.~\textsf{ou-env}), consistent with our theoretical analysis.
The subsequent ODE solving is also efficient for moderate dimensions, with larger $|S_\alpha|$ (e.g., high-order moments in \textsf{gene}) leading to higher computational cost primarily due to the complexity of matrix exponentiation. Note that developing more efficient symbolic solvers for linear ODEs is a complementary and orthogonal direction to our work; in our implementation, we simply rely on off-the-shelf symbolic packages for this step.

\myparagraph{Comparison across models}
Linear and low-dimensional systems (such as the \textsf{consensus} and \textsf{oscillator}) exhibit particularly fast closure and solution times. For nonlinear pro-solvable examples (e.g., \texttt{ou-env} and \texttt{gene}), the closure remains tractable even for moments of degree up to 10, validating the practical scalability of our framework. 
The benchmark \textsf{coupled3d} further shows that our method may still terminate for certain SDEs that \emph{do not} satisfy the pro-solvable property; however, termination is not guaranteed in general.

Overall, the experimental results demonstrate that our approach is effective, broadly applicable to both linear and a wide class of nonlinear systems, and scales well in practice for pro-solvable SDEs. 

\section{Conclusion}\label{sec:conclusion}
We presented a general symbolic method for exact moment estimation of polynomial SDEs, and identified a broad class of pro-solvable systems in which all moments can be computed exactly via finite-dimensional linear ODEs. Both theoretical analysis and experimental results demonstrate its effectiveness and scalability for a wide range of linear and nonlinear models, paving the way for moment-based verification and analysis of stochastic dynamical systems. 

\myparagraph{Limitations} Despite these results, several limitations warrant further discussion. Our method is inherently conditional on the termination of the closure construction. Many polynomial SDEs induce an infinite moment hierarchy, in which case the exact finite-dimensional reduction is unavailable. Even when termination is guaranteed, scalability, particularly with respect to matrix exponentiation, may be limited by the size of the closed moment set. Finally, pro-solvability is sufficient but not necessary for termination.


\myparagraph{Future work} Firstly, when \cref{alg:moment-closure} diverges, one could explore closure approximations by truncation, together with formal error bounds to preserve verification soundness. Secondly, since pro-solvability is not necessary, it would be of interest to characterize termination criteria beyond pro-solvability. Additionally, leveraging sparsity and block structure in the derived ODEs, along with more scalable matrix-exponential techniques, could substantially improve the scalability. 



\bigskip

\myparagraph{Acknowledgments} 
We thank the anonymous reviewers for their valuable comments and helpful suggestions.
This work has been partially funded by the National Key R\&D Program of China under grant No.\ 2022YFA1005101 and 2022YFA1005102, the Open Foundation of Key Laboratory of Cyberspace Security, Ministry of Education of China and Henan Key Laboratory of Network Cryptography under grant No.\ KLCS20240302, the National NSF of China under grant No.\ 62192732, W2511064, and 62502475, the CAS Project for Young Scientists
in Basic Research, and the ISCAS Basic Research under Grant No.\ ISCAS-JCZD-202406.

\myparagraph{Data Availability Statement} 
The artifact and data are available at \url{https://doi.org/10.5281/zenodo.18630506}.

\bibliographystyle{splncs04}
\bibliography{reference}

\clearpage
\appendix
\newpage

\section{Details for Example~\ref{ex:running_example}} \label{appendix:moment-ode}
In~\cref{ex:running_example}, we obtain a closed $8$-dimensional linear ODE system for the collection of moments
\[ m_{(0,2)},\quad m_{(2,1)},\quad m_{(2,0)},\quad m_{(1,1)},\quad m_{(4,0)},\quad m_{(3,0)},\quad m_{(0,1)},\quad m_{(1,0)} \]
where $m_{(i,j)} \defeq \EE[X_t^i Y_t^j]$. The corresponding ODE system is
\begin{align*}
\dot{m}_{(0,2)}(t) &= - 4 m_{(0,2)}(t) + 2 m_{(2,1)}(t) + m_{(2,0)}(t) + 2 m_{(1,1)}(t) \\
\dot{m}_{(2,1)}(t) &= - 4 m_{(2,1)}(t) + m_{(4,0)}(t) + m_{(3,0)}(t) + m_{(0,1)}(t) \\
\dot{m}_{(2,0)}(t) &= 1 - 2 m_{(2,0)}(t) \\
\dot{m}_{(1,1)}(t) &= m_{(2,0)}(t) - 3 m_{(1,1)}(t) + m_{(3,0)}(t) \\
\dot{m}_{(4,0)}(t) &= 6 m_{(2,0)}(t) - 4 m_{(4,0)}(t) \\
\dot{m}_{(3,0)}(t) &= - 3 m_{(3,0)}(t) + 3 m_{(1,0)}(t) \\
\dot{m}_{(0,1)}(t) &= m_{(2,0)}(t) - 2 m_{(0,1)}(t) + m_{(1,0)}(t) \\
\dot{m}_{(1,0)}(t) &= - m_{(1,0)}(t)
\end{align*}
Solving this linear ODE system yields the explicit expression
\[
\mathbb{E}\bigl[Y_t^2\bigr] = m_{(0,2)}(t) =  \frac{1}{3}
+ \frac{2}{3} e^{-3t}
+ \left(-\frac{t}{4} - \frac{11}{8}\right)e^{-2t}
+ \left(\frac{3}{4} t^{2} + t + \frac{3}{8}\right)e^{-4t}\,.
\]

\section{Benchmarks} \label{appendix::benchmarks}

\begin{benchmark}[\textnormal{\textsf{ou-env}~\cite{kallianpur1994stochastic}}]
The system dynamics is the same as in \cref{ex:running_example}: 
\begin{equation}
\begin{cases}
\dif X_t
= - X_t \dif t +  \dif W_t^{(1)}, \\[0.4em]
\dif Y_t
= \bigl(-2 Y_t + X_t +  X_t^{2} \bigr)\dif t
  +  X_t \dif W_t^{(2)},
\end{cases}
\end{equation}
with initial state  $(X_0, Y_0) = (0, 0)$.\qedT
\end{benchmark}

\smallskip

\begin{benchmark}[\textnormal{\textsf{gene}~\cite{sinigh2007stochastic}}]
The system dynamics is: 
\begin{equation}
\begin{cases}
\mathrm{d}X_{1,t}
= \big( - X_{1,t} + 1 \big)\,\mathrm{d}t
  + 0.5\,\mathrm{d}W^{(1)}_t, \\[6pt]
\mathrm{d}X_{2,t}
= \big( 1.2\, X_{1,t} - 0.8\, X_{2,t} \big)\,\mathrm{d}t
  + \big( 0.3\, X_{1,t} + 0.4 \big)\,\mathrm{d}W^{(2)}_t, \\[6pt]
\mathrm{d}X_{3,t}
= \big( 1.0\, X_{2,t} - 0.7\, X_{3,t}
       + 0.2\, X_{1,t}^{2} \big)\,\mathrm{d}t
  + \big( 0.5\, X_{2,t} + 0.1\, X_{1,t}^{2} \big)\,\mathrm{d}W^{(3)}_t, \\[6pt]
\mathrm{d}X_{4,t}
= \big( 0.9\, X_{3,t} - 0.6\, X_{4,t}
       + 0.1\, X_{1,t} X_{2,t} \big)\,\mathrm{d}t
  + \big( 0.4\, X_{3,t} + 0.2\, X_{2,t}^{2} \big)\,\mathrm{d}W^{(4)}_t, \\[6pt]
\mathrm{d}X_{5,t}
= \big( 0.8\, X_{4,t} - 0.5\, X_{5,t}
       + 0.15\, X_{3,t}^{2} + 0.05\, X_{1,t}^{3} \big)\,\mathrm{d}t \\[3pt]
\qquad
  + \big( 0.3\, X_{4,t} + 0.1\, X_{3,t}^{2} + 0.05\, X_{1,t}^{3} \big)
    \,\mathrm{d}W^{(5)}_t .
\end{cases}
\end{equation}
with initial state  $X_{i,0}= 0$ for $i = 1,2,\dots, 5$.\qedT
\end{benchmark}

\smallskip

\begin{benchmark}[\textnormal{\textsf{consensus}~\cite{olfati2007consensus}}]
The system dynamics correspond to those in the first case study, namely the consensus network with noise. Under the specific parameter instantiation considered there, the dynamics are given by
\begin{equation}
\begin{cases}
\mathrm{d}X_{1,t}
=
\bigl(- 2 X_{1,t} + X_{2,t}\bigr)\,\mathrm{d}t
+ \,X_{1,t}\,\mathrm{d}W_t^{(1)},\\[8pt]
\mathrm{d}X_{2,t}
=
\bigl( X_{1,t} - 2 X_{2,t}\bigr)\,\mathrm{d}t
+ \,X_{2,t}\,\mathrm{d}W_t^{(2)}.
\end{cases}
\end{equation}
with initial state  $(X_{1,0}, X_{2,0}) = (1, 0)$.\qedT
\end{benchmark}

\begin{benchmark}[\textnormal{\textsf{vehicles}~\cite{kavathekar2011vehicle}}]
The system dynamics correspond to those in the second case study, namely the vehicle platoon. Under the specific parameter instantiation considered there, the dynamics are given by
\begin{equation}
\begin{cases}
\mathrm{d}p_1 = v_1\,\mathrm{d}t, \\[4pt]
\mathrm{d}v_1 = \bigl(- v_1 + 1\bigr)\,\mathrm{d}t
               + \,\mathrm{d}W_t^{(1)}, \\[4pt]
\mathrm{d}p_2  = v_2\,\mathrm{d}t, \\[4pt]
\mathrm{d}v_2  = \left(-  v_2 + (v_1 - 1)^2 \right)\,\mathrm{d}t 
   + \,\mathrm{d}W_t^{(2)}.
\end{cases}
\end{equation}
with initial state  $(p_1, v_1, p_2, v_2) = (1, 0, 0, 0)$.\qedT
\end{benchmark}

\smallskip

\begin{benchmark}[\textnormal{\textsf{oscillator}~\cite{hafstein2018lyapunov}}]
The system dynamics is: 
\begin{equation}
\begin{cases}
\mathrm{d}X_{1,t}
= X_{2,t}\,\mathrm{d}t,\\[4pt]
\mathrm{d}X_{2,t}
= \bigl(-0.3\,X_{2,t} - X_{1,t}  + 0.8\,X_{3,t}^2\bigr)\,\mathrm{d}t
  + 0.2\,X_{2,t}\,\mathrm{d}W^{(1)}_t,\\[4pt]
\mathrm{d}X_{3,t}
= - X_{3,t}\,\mathrm{d}t + 0.5\,\mathrm{d}W^{(2)}_t,
\end{cases}
\end{equation}
with initial state  $X_{i,0}= 0$ for $i = 1,2, 3$.\qedT
\end{benchmark}



\smallskip

\begin{benchmark}[\textnormal{\textsf{coupled3d}}]
The system dynamics is: 
\begin{equation}
\begin{cases}
\mathrm{d}X_{1,t}
= \bigl(-\tfrac{1}{2} X_{1,t} - X_{1,t} X_{2,t} - \tfrac{1}{2} X_{1,t} X_{2,t}^{2}\bigr)\,\mathrm{d}t
  + X_{1,t}\bigl(1 + X_{2,t}\bigr)\,\mathrm{d}W^{(1)}_{t}\\[4pt]
\mathrm{d}X_{2,t}
= \bigl(-X_{2,t} + X_{3,t}\bigr)\,\mathrm{d}t
  + 0.3\,X_{3,t}\,\mathrm{d}W^{(2)}_{t},\\[4pt]
\mathrm{d}X_{3,t}
= \bigl(X_{2,t} - X_{3,t}\bigr)\,\mathrm{d}t
  + 0.3\,X_{2,t}\,\mathrm{d}W^{(3)}_{t}.
\end{cases}
\end{equation}
with initial state  $(X_{1,0}, X_{2,0}, X_{3,0})= (0, 0, 0)$.\qedT
\end{benchmark}

\section{Extension of Dynkin's Formula}\label{appendix:extension_dynkin}
{\newcommand{\1}{\mathbf{1}}
\newcommand{\R}{\mathbb{R}}
\newcommand{\E}{\mathbb{E}}
\newcommand{\wt}{\wedge}
In this section, we demonstrate that Dynkin's formula also holds for monomials, provided that the moments exist (i.e., they take finite values).

Let $f(x) = \xx^\alpha$ and $\{X_t\}_{t\geq 0}$ be the solution to the SDE. For any $R>0$ define the stopping time
\[
\tau_R := \inf\{t\ge 0:\ ||X_t||\ge R\}.
\]
Let $\chi\in C_c^\infty(\R^n)$ satisfy $0\le\chi\le 1$, $\chi(\xx)=1$ for $||\xx||\le 1$, and $\chi(\xx)=0$ for $||\xx||\ge 2$.
Set $\chi_R(\xx)\defeq \chi(\xx/R)$ and define the compactly supported test function
\[
f_R(\xx)~\defeq~\chi_R(\xx)\,f(\xx)~=~\chi(\xx/R)\,\xx^\alpha.
\]
Then $f_R\in C_c^2(\R^n)$, so the \emph{standard} Dynkin formula (stated for compactly supported test functions) applies to $f_R$. Thus, 
\begin{equation}
\frac{\dif}{\dif t}\E\big[f_R(X_{t\wt\tau_R})\big]
=
\E\big[(\mathcal{A} f_R)(X_t)\big],
\end{equation}
Since $f_R(X_s)=f(X_s)$, $(\mathcal{A} f_R)(X_s)=(\mathcal{A} f)(X_s)$ on the event $\{s<\tau_R\}$, this further simplies to 
\begin{equation}\label{eq:derivative-stopped}
\frac{\dif}{\dif t}\E\big[f(X_{t\wt\tau_R})\big]
=
\E\big[(\mathcal{A} f)(X_t)\,\1_{\{t<\tau_R\}}\big],
\end{equation}
Moreover, since the SDE is polynomial, $A f(X_t)$ is bounded by a polynomial in $X_t$. Assuming the moments of $X_t$ exist (i.e. less than infinity), the dominated convergence theorem justifies taking $R \to \infty$, yielding
\begin{align*}
\frac{\dif}{\dif t}\E\big[f(X_{t})\big]  = \lim_{R\to \infty}
\frac{\dif}{\dif t}\E\big[f(X_{t\wt\tau_R})\big]
& = \lim_{R\to \infty}
\E\big[(Af)(X_t)\,\1_{\{t<\tau_R\}}\big] \\
&= \, \E\big[ \lim_{R\to \infty} (Af)(X_t)\,\1_{\{t<\tau_R\}}\big] = \E\big[(Af)(X_t)\big] 
\end{align*}
Note the assumption that the moments exist is standard can be verified using a ranking supermartingale or a Lyapunov function, as is common in the literature.
}
\end{document}